%% file: main.tex
\documentclass[runningheads]{llncs}
\usepackage{algorithm}
\usepackage{algpseudocode}
\usepackage{graphicx}
\usepackage{textgreek}
\usepackage{amsmath}
\usepackage{xspace}
\usepackage{subfigure}
\usepackage{url}
\usepackage{capt-of}
\graphicspath{ {./images/} }

\begin{document}
\title{Differentially Private Generative Adversarial Networks for Time Series, Continuous, and Discrete Open Data}
\titlerunning{dp-GAN for Time Series, Continuous, and Discrete Open Data, IFIP SEC'19}
%
\author{Lorenzo Frigerio\inst{1} \and
Anderson Santana de Oliveira\inst{2} \and
Laurent Gomez\inst{2} \and Patrick Duverger\inst{3}}
\authorrunning{L. Frigerio et al., IFIP SEC'19}
%
\institute{Polytech Nice-Sophia \and
SAP Labs France, Mougins, France\\
 \and
Ville d'Antibes, Antibes, France\\
}
\maketitle              
\begin{abstract}
Open data plays a fundamental role in the 21th century by stimulating economic growth and by enabling more transparent and inclusive societies.
However, it is always difficult to create new high-quality datasets with the required privacy guarantees for many use cases.
This paper aims at creating a framework for releasing new open data while protecting the individuality of the users through a strict definition of privacy called differential privacy.
Unlike previous work, this paper provides a framework for privacy preserving data publishing that can be easily adapted to different use cases, from the generation of time-series to  continuous data, and discrete data; no previous work has focused on the later class. Indeed, many use cases expose discrete data or  at least a combination between categorical and numerical values.
Thanks to the latest developments in deep learning and generative models, it is now possible to model rich-semantic data maintaining both the original distribution of the features and the correlations between them.
The output of this framework is a deep network, namely a generator, able to create new data on demand.
We demonstrate the efficiency of our approach on real datasets from the French public administration and classic benchmark datasets.

\end{abstract}

\input{introduction}
\input{preliminaries}

\input{problem}
\input{framework}
\input{discrete}
\input{experiments}

\input{relatedworks}
\input{conclusion}

\bibliographystyle{splncs04}

\bibliography{biblio}

\end{document}

%% file: introduction.tex
\section{Introduction}

The digital revolution has changed societies and democracies across the globe, making personal data an extremely valuable asset. In such a context, protecting individual privacy is a key need, especially when dealing with sensitive information, such as political preferences\footnote{Cambridge Analytica scandal has generated huge concerns about how safe our personal data are and what we can do to face this problem: \url{https://www.cnbc.com/2018/03/21/facebook-cambridge-analytica-scandal-everything-you-need-to-know.html}}. 
At the same time, the demand for public administration transparency has introduced guidelines and laws in some countries to release open datasets.

Traditionally, two different  settings  can  be  used  in  privacy  protection;  the  interactive  setting  and the non-interactive one. The interactive setting aims at protecting the privacy of the data while making analyses on them. Indeed, whenever the analyst asks for a query on the data the privacy-preserving mechanism acts directly on the result of the query. Although this approach can target individually the type of query requested, it has several drawbacks. In fact, each new query releases an increasing  amount  of  information,  growing  the  privacy  costs.  Therefore,  only a limited number of queries can be performed on the same dataset. Moreover,specific methods must be designed to manage the diversity of possible analysis that can be made on a database. On the contrary, the non-interactive setting can deal with these problems.
To limit personal data breaches, privacy-preserving data publishing techniques can be employed.  This approach aims at adding the noise directly to the data,  not only to the result of a query (like in interactive settings). The result is a completely new dataset, where analysts can perform an infinite number of requests without increasing the privacy costs, nor disclosing private information. Meanwhile, it  is difficult to preserve the utility of the data.

A strong standard privacy guarantee widely accepted by the research community is differential privacy. It ensures that each individual participating in a database does not disclose any additional information by participating in it. Traditionally, many approaches tried to reach differential privacy by adding noise to the data in order to protect personal information \cite{Geng2016,Dwork2008,Kellaris2014}, however they have never been able to provide satisfying results on real semantically-rich data; most of the implementations were limited to very specific purposes such as histogram queries or counting queries \cite{Xiao2011}. 
Generative models represent the most promising approach in this field. Interesting results have been obtained through Generative Adversarial Networks (GANs) \cite{Goodfellow2014}. These models are able to generate new samples coming from a given distribution. The advantage of generative models is that the noise to guarantee privacy is not added directly to the data, causing a significant loss of information, but it is added inside the latent space, reducing the overall information loss, but guaranteeing meanwhile privacy. 

However, using GANs without an appropriate privacy mechanism may have undesired drawbacks. \cite{Shokri2017} proves that although a GAN should generate unseen samples, in practice, many implementations can be subject to membership inference attacks. \cite{Shokri2017} shows that it is possible to perform attacks with a good accuracy, especially when a model suffers from excessive overfitting. As a consequence, the model is not reliable enough to generate samples with good quality and meanwhile it discloses sensitive information of the training set. 

This paper extends the notion of dp-GAN, an anonymized GAN with a differential privacy mechanism, to handle  continuous, categorical and time-series data. It introduces an optimization called clipping decay that improves the overall performances. This new expansion shapes the noise addition during the training. This allows to obtain a better data utility at the same privacy cost.
A set of analysis on real scenarios evidence the flexibility and applicability of our approach, which is supported by an evaluation of the membership inference attack accuracy, proving the positive effects of differential privacy. We provide experimental results on real industrial datasets from the French public administration and over well-known publicly available datasets to allow for benchmarking.

The remainder of the paper is organized as follows: Section~\ref{sec:background} provides the theoretical background for the paper; 
In Section~\ref{sec:framework}, presents our framework for anonymization, together with the mathematical proofs of differential privacy. 
Section~\ref{sec:experiments} provides a set of experiments on diverse use cases to highlight the flexibility and effectiveness of the approach.
Section~\ref{sec:related} discusses related work; and finally Section~\ref{sec:final} concludes the paper. 

%% file: preliminaries.tex
\section{Preliminaries}
\label{sec:background}
This section brings some important background for the paper.

\subsection{Generative Adversarial Networks.}

GANs (Generative Adversarial Networks) are one of the most popular type of generative models, being already defined as the most interesting idea in the last 10 years in machine learning\footnote{``GAN  and the variations that are now being proposed is the most interesting idea in the last 10 years in ML, in my opinion'', Yann LeCun.}, moreover, a lot of attention has been given to the development of new variations \cite{Chen2017,Mirza2014,JuefeiXu2017}. Given an initial dataset, a GAN is able to mimic its data distribution, for that, a GAN employs two different networks: a generator and a discriminator. The architecture of the two networks is separate from the definition of GAN; depending on the application, different network configurations can be used. The role of the generator is to map random noise into new data instances capturing the original data distribution. On the opposite side, the discriminator tries to distinguish the generated samples from the real ones estimating the probability that a sample comes from the training data rather than the generator. In this way, after each iteration, the generator becomes better at generating realistic samples, while the discriminator becomes increasingly able to tell apart
 the original data from the generated ones. Since the two networks play against each other, the two losses will not converge to a minimum like in a normal training process but this minmax game has its solution in the Nash equilibrium.
Nevertheless, the vanilla GAN \cite{Goodfellow2014} suffers from several issues that make it hardly usable especially for discrete data. A new definition of the loss function, Wasserstein generative adversarial network (WGAN) \cite{arjovsky2017} and Improved Training of Wasserstein GANs \cite{Gulrajani2017} partially solved this problem. We are going to use this latest loss function to train our dp-GAN models.

\subsection{Differential privacy.} The state of the art anonymization technique is differential privacy. This concept ensures that approximately nothing can be learned about an individual whether she participates or not in a database. Differential Privacy defines a constraint on the processing of data so that the output of two adjacent databases is approximately the same. More formally: 
A randomized algorithm $M$ gives ($\epsilon$,$\delta$)-differential privacy if, for all databases $d$ and $d'$, differing on at most one element and all $S \in Range(M)$, \newline
\begin{equation}
Pr[ M(d) \in S] \leq   exp(  \epsilon ) \times Pr[ M(d') \in S] + \delta . \newline
\end{equation}

This condition encapsulates the crucial notion of indistinguishability of the results of a database manipulation by introducing the so-called privacy budget $\epsilon$. It represents the confidence that a record was involved in a manipulation of the database.  Note that the smaller $\epsilon$ is, the more private the output of the mechanism. According to \cite{Dwork2014} the optimal value of $\delta$ is less than the inverse of any polynomial in the size of the database.
Any function $M$ that satisfies the Differential Privacy condition can be used to generate data that guarantees the privacy of the individuals in the database.
In the non-interactive setting, a mechanism $M$ is a function that maps a dataset in another one. The definition states that the probability of obtaining the same output dataset from $M$ is similar, using either $d$ or $d'$ as input for the mechanism.
Composability is an interesting property of Differential Privacy. If $M$ and $M$' are $\epsilon$ and  $\epsilon'$-differential private respectively, their composition $M \circ M'$ is ($\epsilon$+$\epsilon$')-differentially private \cite{Dwork2008}. This property allows to craft a variety of mechanisms and combinations of such mechanisms to achieve differential privacy in innovative ways.

Concerning deep learning, Abadi et al.~\cite{Abadi2016} developed a method to train a deep learning network involving differential privacy. This method requires the addition of a random noise, drawn from a normal distribution, to the computed gradients, to obfuscate the influence that an input data can have on the final model. 
%
%

As for any anonymization methods, one must assess the likelihood of membership inference attacks. This kind of attack evaluates how much a model behaves differently when an input sample is part of the training set rather than the validation set. Given a machine learning model, a membership inference attack uses the trained model to determine if a record was part of the training set or not.  In the case of generative models such as the one of GAN, a high attack accuracy means that the network has been able to model only the probability distribution of the training set and not the one of the entire population. This kind of attack has been proven to be effective especially when overfitting is relevant \cite{Shokri2017}.
%
%

\subsection{Deep Learning with differential privacy}
Abadi et al.\cite{Abadi2016} developed a method to train a deep learning network in a differentially private manner. This method requires the addition of a random noise to the computed gradients in order to obfuscate the influence that an input data can have on the final model. The training process can be defined as follow:

\begin{itemize}
\item A set of input, called lot, is processed by the model and the loss is computed. 
\item The array composed of the gradients for each weight is calculated starting from the loss.
\item The gradients are clipped, returning the minimum value between the norm of the gradients and an upper bound defined as an hyperparameter of the model.
\item A noise coming from a normal distribution is added to the clipped gradients with a variance proportional to the upper bound used before. 
\item The result is a sanitized version of the gradients in which the influence of the input data is bounded, guaranteeing privacy. These new gradients will be used to train the model updating the weights.
\end{itemize}

\subsection{Membership inference attacks}
A powerful attack that affects most of the machine learning algorithms is represented by membership inference attack. This kind of attack evaluates how much a model behaves differently when an input sample is part of the training set rather than the validation set. Given a machine learning model, a membership inference attack uses the trained model to determine if a record was part of the training set or not. Consequently, it analyses if the model is putting in danger the privacy of the data present in the training set. In the case of generative models such as the one of GAN, an high attack accuracy means that the network has been able to model only the probability distribution of the training set and not the one of the entire population. This kind of attack has been proven to be effective especially when overfitting is relevant \cite{Shokri2017}.
When we have the trained model at our disposal, it is possible to apply membership inference attacks through white box attacks. In the case of GAN this process is simple but at the same time effective. In practice, a new dataset made by the combination of training and validation samples is passed to the discriminator network. The discriminator outputs for each sample a value that represents the probability that the data was part of the training set. If the model has correctly performed generalization, the discriminator should not be able to distinguish the training samples from the validation ones because they come from the same distribution. On the contrary, if the training samples have higher probabilities, it means that the model has not learned the distribution of the training set, but it has only learned to reproduce those specific samples without generalization.

%% file: problem.tex
\section{Problem Statement}
Database records in ERP or CRM systems are mostly categorical and discrete data, while several IoT applications produce time series data. In the context of the growing need for public administrations to release open dataset for increased transparency, the open questions the current paper addresses are: 
\begin{itemize}
\item how to provide a differentially private GAN framework flexible enough to handle not only continuous dataset, but also time series and discrete data?
\item What are the metrics to employ to search for a stable state providing suitable values for $\epsilon$ and $\delta$?
\end{itemize}
The targeted use case is the release of a dataset in an anonymized manner so that can safely be used as open data. Governments have the necessity to increase transparency through the publication of new open datasets. However, even more severely after the introduction of new global data protection regulations, the rules are becoming stricter. Anonymization is considered a suitable way to guarantee the privacy of the users participating in a dataset and therefore favor the creation of new open data. 

%% file: framework.tex
\section{The framework}
\label{sec:framework}

In our framework proposed we assume a trusted curator interested in releasing a new open dataset with privacy guarantees to the users present in it. Outside the trusted boundary, an analyst can use the generator model, result of our algorithm, to perform an indefinite number of queries over the data the generator produces. Such outputs can be eventually released as open data. Even by combining the generated data with other external information, without ever having access to the original training data, the analyst would not be able to violate the privacy of the information, thanks to the mathematical properties of differential privacy.

The dp-GAN model is constituted by two networks, a generator and a discriminator, that can be modelled based on the application domain. We adopted Long Short Term Memories (LSTM) inside the generator to model streaming data and multilayer perceptron (MLP) to model discrete data. In addition, to manage discrete data, we also used a trick that does not influence the training algorithm but it changes the architecture of the generator network. Specifically, an output is created for each possible value that a variable can assume and a softmax layer is added for each variable. The result of the softmax layer becomes the input of the discriminator network. Indeed, each output represents the probability of each variable instance; the discriminator compares these probabilities with the one-hot encoding of the real dataset. On the contrary, the output nodes associated with continuous variables are kept unchanged.

At the end of the training, the generator network can be publicly released; in this way, the analyst can generate new datasets as needed. Moreover, since the generator only maps noise into new data the process is really fast and data can be generated on the fly when a new analysis is required.

We used the differentially private Stochastic Gradient Descent (dp-SGD) proposed by \cite{Abadi2016} to train the discriminator network and the Adam optimizer to train the generator. The dp-GAN implementation relies on a traditional training in which the gradients computed for the discriminator are altered. This due to the fact that we want to limit the influence that each sample has on the model. On the contrary, the training of the generator remains unaltered; indeed, this network bases its training only on the loss of the discriminator without accessing directly the data.

The approach is independent from the chosen loss and therefore can be applied to the vanilla GAN implementation~\cite{Goodfellow2014} but also to the improved WGAN one. The dp-SGD works as follows: once the gradients are calculated, it clips them by a threshold $C$ and alter them by the addition of a random noise with variance proportional to $C$. Each time an iteration is performed, the privacy  cost increases and the objective is to find a good balance between data utility and privacy costs. 

Our implementation is an extension to the improved WGAN framework combining it with the dp-SGD. Therefore, the loss functions are calculated as in a normal WGAN implementation, except that the computed gradients are altered to guarantee privacy. Moreover, for the first time up to our knowledge, the dp-GAN concept is adapted to handle discrete data. Algorithm~\ref{algo} describes our training procedure.
\begin{algorithm}
  \caption{Algorithm for training a GAN in a differentially private manner \label{algo}}

  \begin{algorithmic}
    \Statex  { \textbf{Input}: Samples from $x_1$ to $x_N$, group size $L$, number of samples $N$, clipping parameter $C$, noise scale \textsigma, privacy target \textepsilon, number of iterations of the discriminator per each iteration of the generator $Ndisc$, batch size $b$, Wasserstein distance $\mathcal{L}$, learning rate \texteta,  number of discriminator' s parameters $m$, clipping decay $C_{decay}$. 
\newline \textbf{Output}: differentially private Generator G
}
    \State Initialize weights randomly both for the Generator $\theta_{G(0)}$ and the discriminator $\theta_{D(0)}$
    \State Convert discrete variables into their One-Hot encodings
    \While{ (While privacy cost $\le$ \textepsilon ) } 
    \For {$t=0$ to $Ndisc$}
    \For {$j=0$ to $b$}    
    \State  sample $L\textsubscript{t}$ with sample probability $L/N=q$
    \State For each  $x_i$ in $L\textsubscript{t}$, compute $g\textsubscript{t}(x\textsubscript{i}) \gets \nabla_\theta \mathcal{L}(\theta\textsubscript{t},x\textsubscript{i})$    
    \Comment{Compute gradient}  
    \State $g\textsubscript{t}(x\textsubscript{i}) \gets g\textsubscript{t}(x\textsubscript{i})/max(1, \left\lVert g\textsubscript{t}(x\textsubscript{i})\right\rVert / C)$
    \Comment{Clip gradient} 
    \State $g\textsubscript{t} \gets \frac{1}{L} (\sum\limits_{i=0}^L  g\textsubscript{t}(x\textsubscript{i}) +N(0,(\sigma * C)^2 I)) $
    \Comment{Add noise} 
    \State $\theta\textsubscript{D(t+1)}\gets \theta\textsubscript{D(t)} - \eta * g\textsubscript{t}$
    \Comment{Gradient descent}
    \EndFor
    \EndFor
    \State $C$ *=  $C_{decay}$
    \Comment{Clipping decay}       
    \State Update the overall privacy cost \textepsilon     
    \Comment{Moment accountant} 
	\State Sample m values {z\textsubscript{i}} $\sim$ Random noise
    \Comment{Sample random noise}  
    \State $\theta_{G(t+1)} \gets  Adam(\nabla_\theta \frac{1}{m} \sum\limits_{i=0}^m -D(G(z\textsubscript{i}))) $
    \Comment{Update Generator}  
    \EndWhile 
    \State {\bf return} G
  \end{algorithmic}
\end{algorithm}
\subsection{Clipping decay}
The role of the clipping parameter is to limit the influence that a single sample can have on the computed gradients and, consequently, on the model. Indeed, this parameter does not influence the amount of privacy used. A big clipping parameter allows big gradients to be preserved at the cost of a noise addition with a proportionally high variance. On the contrary, a small clipping parameter limits the range of values of the gradients, but it keeps the variance of the noise small. The bigger the clipping parameter the bigger the gradients' variance. Similarly to what it is done with the learning rate, it is possible to introduce a clipping parameter decay. In this way, the gradients not only tend to descend over time to better reach a minimum but, in addition, they mimic the descending trend of the gradients allowing to clip the correct amount at each step. In fact, when the model tends to converge to the solution, the gradients decrease. Therefore, the noise may hide the gradients if its variance is kept constant. By reducing the clipping parameter over time, it is possible to reduce the variance in the noise in parallel with the decrease of the gradients, thus improving the convergence of the model. This without influencing the overall privacy costs that are not altered by the clipping parameter but only by the amount of noise added.

\subsection{Moment accountant}
A key component of the dp-GAN is the moment accountant. It is a method that allows to compute the overall privacy costs by calculating the cost of a single iteration of the algorithm and cumulating it with the other iterations. Indeed, thanks to the composability property of differential privacy it is possible to cumulate the privacy costs of each step to compute the overall privacy cost. 
Given a correct value of $\sigma$ and thanks to  weights clipping and the addition of noise, Algorithm~\ref{algo} is $(O(\epsilon,\delta))$-DP with respect to the lot. Since each lot is sampled with probability $q = L/N$, each iteration is $(O(q\epsilon,q\delta))$-DP. In the formula, q represents the sampling probability (the number of samples inside a lot divided by the total number of samples present in the dataset). The clipping decay optimization has no influence on the moment accountant. Indeed, it alterates only the clipping parameter and not the variance of the noise that is the variable that influences the cost of a single iteration by changing the value of $\epsilon$.
Each time a new iteration is performed the privacy costs increase. However, thanks to the definition of moment accountant, these costs do not increase linearly.
Indeed, by cumulating the privacy costs for each iteration, an overall level of $(O(q\epsilon \sqrt{T}),\delta)$-DP is achieved where $T$ represents the number of steps (the number of epochs divided by $q$).

The moment accountant is based on the assumption that the composition of Gaussian mechanisms is being used. Assessing that a mechanism $M$ is $(O(\epsilon,\delta))$-DP is equivalent to a certain tail bound on $M$'s privacy loss random variable. The moment accountant keeps track of a bound on the moments of the privacy loss random variable defined as: \newline
\begin{equation}
c(o;M,aux,d,d') = log \frac{Pr[M(aux,d) = o]}{Pr[M(aux,d') = o]}\newline \label{lossrandomvariable}
\end{equation}
In (\ref{lossrandomvariable}) $d$ and $d$' represent two neighbouring  databases, $M$ the mechanism used, $aux$ an auxiliary input and $o$ an outcome. What we are computing are the log moments of the privacy loss random variable that can be cumulated linearly. 
In order to bound this variable, since the approach is the sequential application of the same privacy mechanism we can define the $\lambda \textsuperscript{th}$ moment $\alpha M(\lambda, aux, d, d')$ as the log of the moment generating function evaluated at the value \textlambda : 
\begin{equation}
 M(\lambda,aux,d,d') = log E\textsubscript{o $\sim$ M(aux,d)} [exp(\lambda c(o,M,aux,d,d'))].
\end{equation}
And consequently we can bind all possible $\alpha M(\lambda,aux,d,d')$. We define \begin{equation}\label{eq:mech}
\alpha M(\lambda) = max_{aux,d,d'} \alpha M(\lambda,aux,d,d') 
\end{equation}
\begin{theorem} \label{th1}
Using the definition \eqref{eq:mech} then $\alpha M(\lambda)$ has the following characteristics:
given a set of $k$ consecutive mechanisms, for each $\lambda$: \[ \alpha_M(\lambda) \le \sum\limits_{i=1}^k \alpha M_i(\lambda) \]
for any $\epsilon > 0$, a mechanism $M$ is $(\epsilon,\delta)$-differentially private for
\[ \delta = min_\lambda exp(\alpha_M(\lambda) - \lambda * \epsilon) \]
\end{theorem}
\begin{proof}[Proof of Theorem \ref{th1}]
A detailed proof of Theorem 1 can be found in \cite{Abadi2016}.
\end{proof}
\begin{theorem} \label{th2}
 Algorithm~\ref{algo} is $(O(q\epsilon \sqrt{T}),\delta)$-differentially private for appropriately\footnote{The appropriate values for the noise scale and for the threshold will depend on the desired privacy cost and on the size of the dataset.} chosen settings of the noise scale and the clipping threshold.
\end{theorem}
\begin{proof}[Proof of Theorem \ref{th2}]
By Theorem~\ref{th1}, it suffices to compute, or bound, $\alpha M_i(\lambda)$ at each step and sum them to bound the moments of the mechanism overall. Then, starting from the tail bound we can come back to the  $(\epsilon, \delta)$-differential privacy guarantee.
The last challenge missing is to bind the values $\alpha M_t(\lambda)$ for every single step. Let $\mu_0$ denote the Probability Density Function (PDF) of $N(0,\sigma^2)$, and $\mu_1$ denote the PDF of $N(1,\sigma^2)$. Let $\mu$ be the mixture of two Gaussians $\mu = (1 - q)\mu_0 + q\mu_1$. Then we need to compute $\alpha(\lambda) = log (max(E_1, E_2))$ where
\begin{equation}
E_1 = E\textsubscript{z}  [( \mu\textsubscript{0}(z) / \mu(z))^\lambda]
\end{equation}
\begin{equation}
E_2 = E\textsubscript{z}  [( \mu(z) / \mu\textsubscript{0}(z))^\lambda]
\end{equation}
\newline In the implementation of the moment accountant, we carry out numerical integration to compute \textalpha(\textlambda). In addition, we can show the asymptotic bound \[\alpha(\lambda) \le q^2 \lambda(\lambda + 1)/(1 - q)\sigma^2 + O(q^3 /\sigma^3 )\]
\newline
This inequation together with Theorem~\ref{th1} implies Theorem~\ref{th2}.
\qed
\end{proof}

%% file: discrete.tex
\section{Discrete Scenario}
In this paper we target also discrete data, so, we extended the dp-GAN implementation to manage datasets in which both continuous and categorical data are present. The trick used, does not influence the training algorithm while it changes the architecture of the generator network. Specifically, an output is created for each possible value that a variable can assume and a softmax layer is added for every variable. The result of the softmax layer is the input for the discriminator network. Indeed, each output constitutes the probability of each variable instance; these probabilities are compared by the discriminator with the one-hot encoding of the real dataset. On the contrary, continuous variables are kept unchanged. After training the network, the new dataset is generated by sampling a value based on the probabilities coming from the softmax layer, for each variable. Differential privacy influences these probabilites by introducing randomness: this prevents outliers from influencing analyses carried out on the anonymized dataset too much.

\subsection{Improvements}
The role of the clipping parameter is to limit the influence that a single sample can have on the computed gradients and, consequently, on the model. Indeed, this parameter does not influence the amount of privacy used. A big clipping parameter allows big gradients to be preserved at the cost of a noise addition with a proportionally high variance. On the contrary, a small clipping parameter limits the range of values of the gradients, but it keeps the variance of the noise small. The bigger the clipping parameter the bigger the gradients' variance. The clipping parameter acts in parallel to the learning rate and consequently, these two hyperparameters should be tuned together. Similarly to what it is done with the learning rate, it is possible to introduce a clipping parameter decay. In this way, the gradients not only tend to descend over time to better reach a minimum, but, in addition, they mimic the descending trend of the gradients allowing to clip the correct amount at each step. in fact, when the model tends to converge to the solution, the gradients decrease. Therefore, the noise may hide the gradients if its variance is kept constant. By reducing the clipping parameter over time, it is possible to reduce the variance in the noise in parallel with the decrease of the gradients, thus improving the convergence of the model. This without influencing the overall privacy costs that are not altered by the clipping parameter but only by the amount of noise added. This intuition is experimentally proved in the experiments Section.

%% file: experiments.tex
\section{Experiments}
\label{sec:experiments}
In this section, we evaluate empirically our framework. The experiments are designed to assess the quality of the generated data, measure the privacy of the generated models and understand how differential privacy influences the output dataset. Moreover, we evaluate the solidity of the different models against membership inference attacks. 
Since it is notoriously arduous to assess the results of a GAN, we decided to combine qualitative  and quantitative analysis to obtain a reliable evaluation. 
Qualitative analysis allows us to graphically verify the quality of the results and to observe the effects of differential privacy; while quantitative analysis provides a most accurate evaluation of the results; in particular, we measured some distance metrics to compare the generated data with the real data. Finally, our process included evaluating  our model on a classification problem. This highlights the high utility of the data even when anonymization is used. 
For all experiments, when differential privacy is used, $\delta$ is supposed to be less than $10^{-5}$, a value that is generally considered 
 safe \cite{Abadi2016} because it implies that the definition of differential privacy is true with a probability of $99.999 \%$. Indeed $\delta$ is the probability that a single record is safe and not spoiled. We kept the value of $\delta$ fixed to be able to evaluate the privacy of a mechanism with a single value $\epsilon$ that summarizes in a clearer manner the privacy guarantees.

In the different settings we applied only minor changes to the dp-GAN architecture, since we proved that it adapts well to each of them. In particular, in every case the discriminator is composed of a deep fully connected network. On the other hand, the architecture of the generator is adapted to the different datasets used.
To generate time-series we used an LSTM which output becomes the input of the discriminator. On the contrary, in the case of discrete datasets we used a fully connected network which outputs the probability distribution for each value that a variable can assume.
The interested reader can find an exhaustive explanation of the experiments, including additional datasets in the following GIT repository: \url{https://github.com/Lory94/dp-GAN}.

\subsection{Synthetic dataset}
In order to provide a first evaluation of the performances of the dp-GAN and understand the effects of differential privacy, we conducted a first experiment on a synthetic dataset. The dataset is constituted by samples coming from six 2D-gaussian distributions with the same variance, but with different centers. 
The quality of the results using dp-GAN is similar in both marginals and joint distributions. 

\begin{figure*}[!hb]
    \centering
    \minipage{0.32\textwidth}
        \centering
        \begin{subfigure}[]{}
           \includegraphics[width=\linewidth]{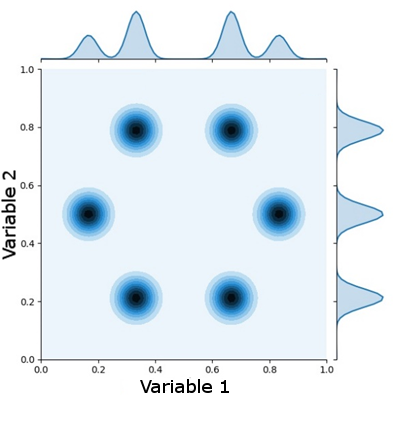}
        \end{subfigure}%
    \endminipage
    ~ 
    \minipage{0.32\textwidth}
        \centering
        \begin{subfigure}[]{}            \includegraphics[width=\linewidth]{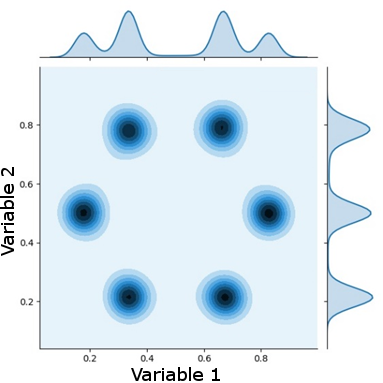}
        \end{subfigure}
   \endminipage
    ~
    \minipage{0.32\textwidth}
        \centering
        \begin{subfigure}[]{}            \includegraphics[width=\linewidth]{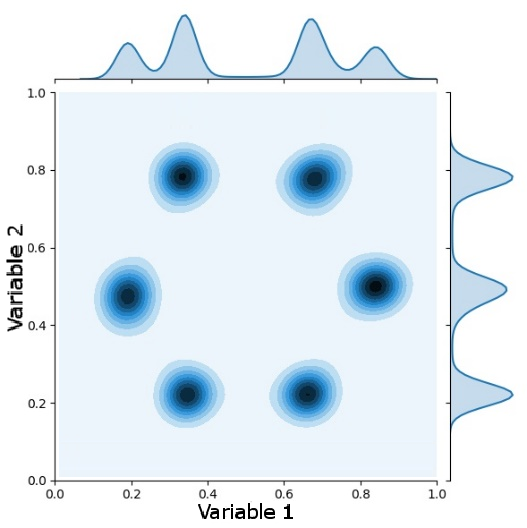}
       \end{subfigure}
    \endminipage
    \caption{Kernel estimation for: (a) the original points, (b) WGAN, and  (c) dp-GAN }
    \label{fig:kernels}
\end{figure*}

Fig.~\ref{fig:kernels}  plots the kernel density estimation of the data to visualize the bivariate distribution. As expected, differential privacy introduces a small noise in the results, thus increasing the variance of the six gaussian distributions, while at the same time replicating faithfully the original distribution.


We use the Wasserstein distance to measure the distance between two distributions, to ascertain the quality of the GAN models. 
Fig.~\ref{fig:synthetic_was} plots the distance values for the non-anonymized GAN, a dp-GAN,  and a dp-GAN using clipping decay. Both the dp-GAN models have $\epsilon = 8$. The different measures tend to converge to similar results, especially when clipping decay is applied, demonstrating the high quality of the results. Indeed, clipping decay allows the Wasserstein distance to drop to values comparable to those of the non-anonymized version in the second half of the graph. The main difference resides in the higher number of epochs necessary to reach the convergence, due to the noise addition.



\subsection{Time-series data}
To test our implementation on a real dataset we decided to use a set of data coming from the IoT system of the City of Antibes, in France. This dataset is private because it contains sensitive information about the water consumption and water pipeline maintenance, obtained directly from sensors in each neighborhood. The purpose is to support public administration in releasing highly relevant open data, while hiding specific events in the time-series, and preventing individual re-identification.  With minor adjustments, the solution can constitute a valid framework applicable to other purposes, such as electricity consumption and waste management. 

The dataset is an extract of one month of measurements, where each sample is a time-series containing 96 values (one every 15 minutes). Each sample is labelled with the name of the neighborhood. The data has been normalized before the training. The  goal is to generate a new time-series that contains the same number of records and the same distribution as the original dataset, while providing differential privacy guarantees, that is, each sample does not influence whatever analyses more than a certain threshold. In this way, anomalous situations such as maintenance works, a failure in a water pipe or an unexpected water usage by a person living in a certain area are protected and kept private. 

\begin{figure*}[!htb]
    \centering
    \minipage{0.33\textwidth}
        \centering
            \includegraphics[width=\linewidth]{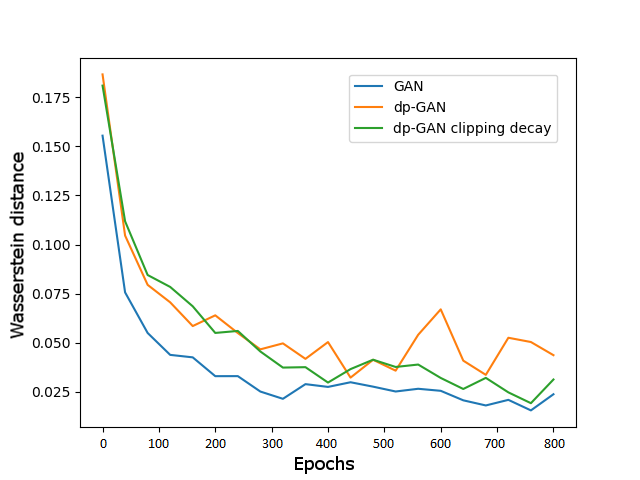}\caption{ Wasserstein distance using anonymized and non anonymized GANs
    }\label{fig:synthetic_was}
        
    \endminipage \hfill
    \centering
    \minipage{0.66\textwidth}
    \minipage{0.5\textwidth}
        \centering
        \begin{subfigure}[]{}
            \includegraphics[width=\linewidth]{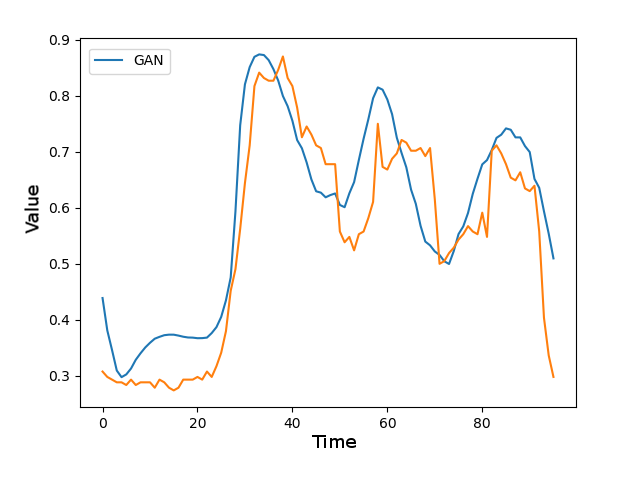}
        \end{subfigure}%
    \endminipage \hfill
   \minipage{0.5\textwidth}
        \centering
        \begin{subfigure}[]{}
          \includegraphics[width=\linewidth]{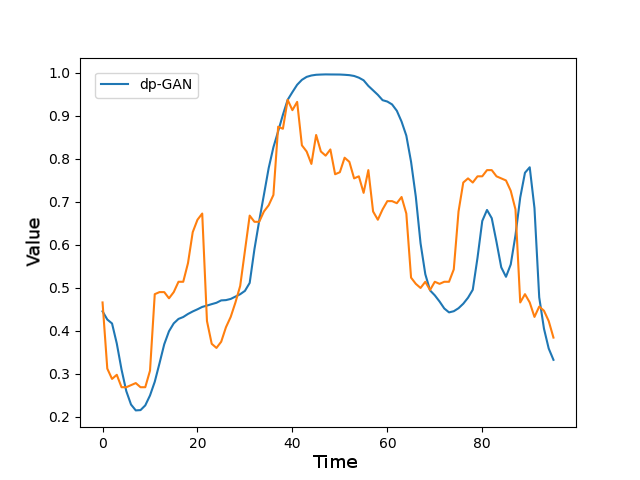}
       \end{subfigure}
    \endminipage \hfill \caption{Generated sample from a non-anonymized GAN (a) and dp-GAN (b), in blue. In orange, the closest sample present in the dataset, in terms of dynamic time warping. }
    \label{fig:antibes}
    \endminipage \hfill
   
\end{figure*}

Fig.~\ref{fig:antibes} compares real samples and generated ones, using non-anonymized-GAN and a dp-GAN with $\epsilon = 6$. We plot the sensor values for  a generated sample and the closest sample coming from the original data, in terms of the dynamic time warping distance. The distribution of the original time series is kept, but in the dp-GAN samples, the curves tend to be smoother, hiding some of the variability of the original data. 

For time-series data, the quality assessment for GANs represents a challenge. While for images the inception score \cite{luvcic2017gans} has become the standard measure to evaluate the performance of a GAN, there is no counterpart for the assessment of time series. We believe that this represents an interesting area of research for the future.


\subsection{Discrete data}
We analyzed the performances of our model on the UCI adult dataset: this dataset is an extract of the US census and contains information about working adults, distributed across $14$ features. A classification task on it is a reliable benchmark, because of its widespread use in several studies.
Records are classified depending on whether the individual earns more or less than $50k$ dollars each year. 

To use accuracy as an evaluation metric, we decided to sample the training and test data in such a way that both classes would be balanced. 
We built a random forest classifier on the dataset generated by the dp-GAN. We evaluated the accuracy on the test set and compared it with the one of the model built on the real non-anonymized dataset. 
If the dp-GAN model behaves correctly, all the correlations between the different features should be preserved. Therefore, the final accuracy should be similar to what was achieved by using the real training set. We also tracked the privacy costs to verify that the generated data were correctly anonymized. Finally, we examined how much membership inference attacks can influence our model and compared it to a non-anonymized GAN model.

\begin{table}[ht]
\minipage{0.49\textwidth}
\centering
\begin{tabular}{ l| c|c }
  \hline			
  Method & Epsilon & Accuracy \\
  Real dataset & Infinite & 77.2 \% \\
  GAN & Infinite & 76.7 \% \\
  dp-GAN & 3 & 73.7 \% \\
  dp-GAN clipping decay & 3 & 75.3 \% \\
  dp-GAN & 7 & 75.0 \% \\
  dp-GAN clipping decay & 7 & 76.0 \% \\	[1ex]
  \hline  
\end{tabular}
\vspace{1mm}
\caption{Classification accuracy for training sets generated by different models}
\label{table:1}
\endminipage
\centering
\minipage{0.49\textwidth}
  \includegraphics[width=\linewidth]{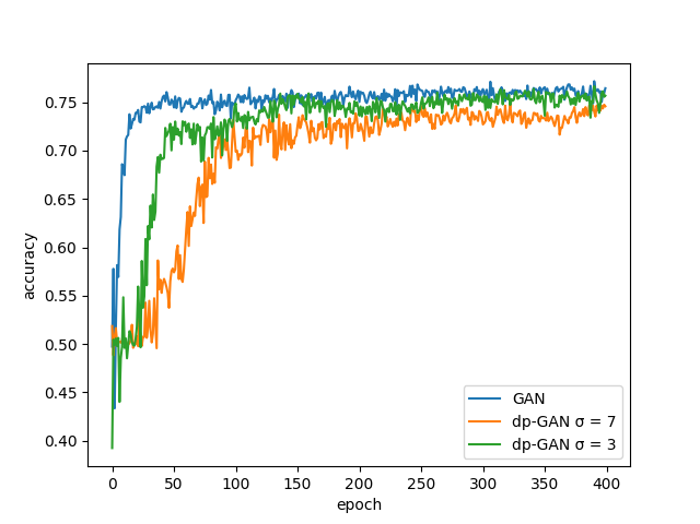}
  \captionof{figure}{classification accuracy Average for $5$ runs for different noise values}
  \label{fig:awesome_imag}
\endminipage\hfill
\end{table}

\begin{figure}[!htb]
\minipage{0.49\textwidth}
  \includegraphics[width=\linewidth]{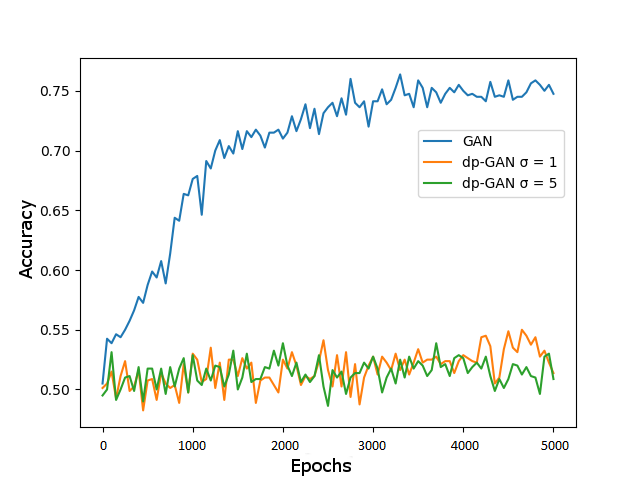}
  \caption{Membership inference attack accuracy for non-anonymized GAN, dp-GAN with $\sigma = 1$ and $\sigma = 5$ }\label{fig:awesome_image9}
\endminipage\hfill
\minipage{0.49\textwidth}
  \includegraphics[width=\linewidth]{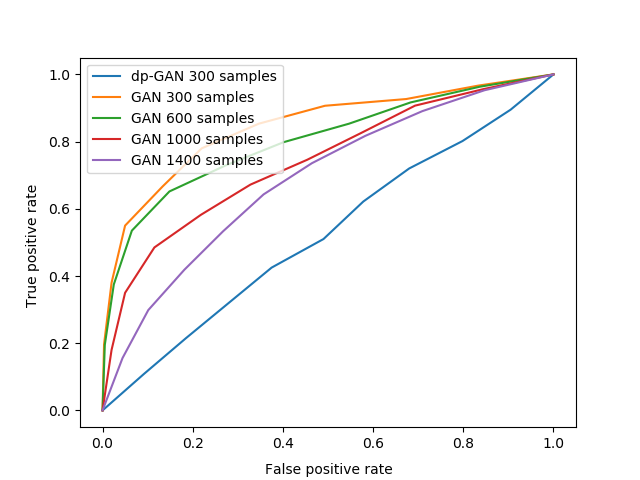}
  \caption{ROC curves for membership inference attacks for GAN and dp-GAN using generated samples of different size}\label{fig:roccurves}
\endminipage
\end{figure}
Table~\ref{table:1} evaluates the degradation of the performances when differential privacy is adopted. The target accuracy of $77.2\%$ was reached by using the real training set to train the random forest classifier. As expected, a non-anonymized GAN is able to produce high quality data: its accuracy loss is very low, $0.5\%$, compared to the target. Interestingly, even when we adopted the dp-GAN framework the accuracy remained high. Using $\epsilon = 7$ and  clipping decay, we obtained results similar to the ones without anonymization. In addition,  Table. \ref{table:1} points out the positive effect of clipping decay, that it is able to increase the accuracy of about $1\%$.

Fig.~\ref{fig:awesome_imag} highlights the effects on the classification accuracy when dp-GAN is adopted using different amounts of noise. The main effect is to slow down the training process, but not to significantly impact accuracy. Indeed, the added noise requires more epochs to reach convergence, which is  amplified by $\sigma$. In most of the use cases, this is a minor drawback, considering the classification accuracy. Moreover, the dp-GAN is trained once;   then the generator can be released to produce samples on demand. 

Fig.~\ref{fig:awesome_image9} shows the analysis of the accuracy of membership inference attacks on the model using different training procedures with different levels of privacy guarantees. This analysis has been done at different epochs of the training process. The accuracy of the model increases over time, however this makes the model more subject to membership inference attacks. As can be seen from  Fig.~\ref{fig:roccurves}, training the model with no anonymization rapidly increases the accuracy of attacks: this highlights the problems that still afflict many generative models that cannot effectively generalize the training data. On the contrary, by increasing the privacy level, the accuracy of the attacks tends to remain close to $50\%$. This is obtained at the costs of losing about $1\%$ of accuracy during the final classification.


The size of the dataset is another important factor that influences significantly the results, since GANs need a good amount of data to generalize effectively. Fig.~\ref{fig:roccurves} confirms the results obtained in \cite{Shokri2017}, but at the same time it shows how
differential privacy works well even when the dataset is small. Indeed, the dp-GAN provides random accuracy towards membership inference attacks
independently from the size of the dataset. It is interesting to notice that since the dataset is small, the level of privacy $\epsilon$ is big compared to what it is commonly used; however, the effects of differential privacy can be still perceived clearly.

\subsection{Clipping decay}
In order to test our model on a fully discrete scenario in which analyse the effects of clipping decay on differential privacy more easily we employed the UCI mushrooms dataset.
 This dataset is composed by 22 categorical features such as odor or habitat and can be used in a classification task. The classification task consists of classifying each mushroom to one of two different categories: edible or
 poisonous. Given the dataset, we split it in training set and test set and we anonymized it computing the accuracy using different levels of privacy as already done for the UCI adult dataset. Fig.8 shows the effects of differential 
 privacy on the classification. It is clear that also the model using differential privacy with clipping decay can reach comparable results as the non-anonymized GAN; however it requires more time to converge due to the 
 addition of noise. On the contrary, the dp-GAN without clipping decay is able to reach similar but not equal results. Indeed, clipping decay slows down the training at the beginning, but in the long term it is able to obtain
 better results. The two dp-GAN use the same amount of privacy (because they use the same $\sigma = 7$) and the only difference is the capacity of the noise to adapt to the different stages of the training.
\begin{figure}[!htb]
  \includegraphics[width=0.60\linewidth]{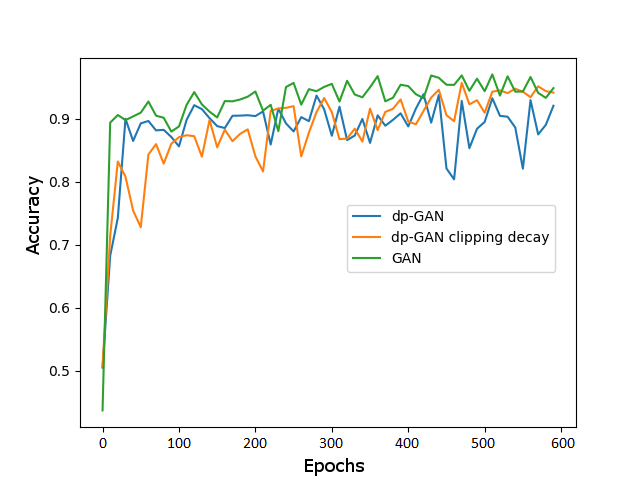}
  \caption{Comparison of the accuracies measured using different levels of privacy on the UCI mushrooms dataset.}\label{fig:awesome_image10}
\end{figure}

%% file: relatedworks.tex
\section{Related work}
\label{sec:related}

\textbf{Differential privacy on machine learning models}. In \cite{Papernot2016} it is proposed an innovative approach for training a machine learning model in a differentially private manner. On the contrary of the dp-SGD, they proved that it is possible to reach differential privacy by transferring knowledge from some models to others in a noisy way. A set of models, called teachers, are trained on the real dataset and a student model learns in a private way what the teachers have grasped during the training. However, it is still unclear how this implementation can be extended to a non-interactive setting. \cite{Xie2018} developed a dp-GAN based on the dp-SGD providing some optimizations in order to improve performances focusing on the generation of images. In contrast, our work highlights that dp-GAN can be adapted to a variety of different use cases and in particular we developed a variation dedicated to discrete data. In addition, we provide, also, an overview of the effects that differential privacy has on membership inference attacks; \cite{Shokri2017} pointed out how severe the risk of this kind of attack in a general machine learning model can be. We have confirmed the issue while highlighting that the noise introduced by differential privacy reduces overfitting and consequently the accuracy of membership inference attacks.
\newline \textbf{Generative adversarial networks on discrete data}. \cite{Yu2016} developed SeqGAN, an approach dedicated to the generation of sequences of discrete data. SeqGAN is based on a reinforcement learning training where the reward signal is produced by the discriminator. However, it is not clear how this approach can be extended to include differential privacy. On the contrary, \cite{Mottini2018AirlinePN} uses Cramer GANs to combine discrete and continuous data. These recent works did not address data privacy concerns.
\newline \textbf{Differential privacy without deep learning}.
Interesting results have been also obtained through other types of generative models. In the context of non-interactive anonymization,
\cite{Zhang2017} developed a differentially private method for releasing high-dimensional data through the usage of Bayesian networks. This kind of network is able to learn the correlations present in the dataset and generate new samples. In particular, the distribution of the dataset is synthetized through a set of low dimensional marginals where noise is injected and from which it is possible to generate the new samples.
However, the approach suffers from an extremely high complexity, thus being unpractical to anonymize large datasets.
We have also analysed the literature to prove that the amount of privacy that dp-GAN guarantees is comparable to the one of the other most common implementations. Although there is no specific value for which $\epsilon$ is considered safe, we obtained most often lower privacy costs compared to ~\cite{erlingsson2014rappor,Apple2017}, which are the two most relevant works dealing with real-life datasets.  Similar privacy costs have been 
used in the most recent literature in the differential privacy field \cite{Abadi2016,Papernot2016}.

%% file: conclusion.tex
\section{Conclusion}
\label{sec:final}
In this paper, we apply a novel approach in the release of anonymized data. The privacy guarantee that we used, differential privacy, is the state of the art in this field and it provides adequate protection from very sophisticated attacks, such as linkage attacks; this characteristic cannot be guaranteed in other common anonymization methods such as k-anonymity. 
We propose an extension to the previous dp-GAN implementations that targets discrete data and we introduce a major contribution, clipping decay, to optimize performances.
We assessed that this approach can be adopted in highly diverse real case scenarios while maintaining a good utility but at the same time preserving privacy especially when treating large datasets. We have also pointed out the beneficial effects that this implementation provide against membership inference attacks. The primary goal of this paper is to provide an effective approach towards releasing new open data. Indeed, although open data constitute an important value for many companies and governments, it is often difficult to produce new, high-quality datasets. In the future we will continue to work towards the reduction of the amount of privacy used through an additional set of optimizations and we are going to assess what kind of benefits can transfer learning provide to data anonymization.